\documentclass[11pt, a4paper]{article}

\usepackage[utf8]{inputenc}
\usepackage[T1]{fontenc}
\usepackage[english]{babel}
\usepackage{amsmath,amssymb,amsthm,mathtools}
\usepackage{bm}
\usepackage{textcomp}
\usepackage{enumitem}
\usepackage{geometry}
\usepackage{graphicx}
\usepackage{tikz-cd}
\usepackage{xcolor}
\usepackage{url}
\DeclareGraphicsExtensions{.pdf,.png,.jpg}
\graphicspath{{figures/}}
\usepackage[colorlinks=true,allcolors=blue]{hyperref}
\usepackage[capitalize]{cleveref}
\usepackage[numbers,sort&compress]{natbib}
\makeatletter

\newcommand{\Rmnum}[1]{\expandafter@slowromancap\romannumeral #1@}
\makeatother

\renewcommand{\d}{\,\mathrm{d}}

\newcommand{\E}{\mathbb{E}}
\newcommand{\Pp}{\mathbb{P}}

\newcommand{\free}{\mathrm{free}}
\newcommand{\fix}{\mathrm{fix}}

\newcommand{\Hh}{\mathbb{H}}
\newcommand{\Z}{\mathbb{Z}}
\newcommand{\dist}{\operatorname{dist}}

\theoremstyle{plain}
\newtheorem{theorem}{Theorem}[section]
\newtheorem{lemma}[theorem]{Lemma}
\newtheorem{proposition}[theorem]{Proposition}
\newtheorem{corollary}[theorem]{Corollary}
\theoremstyle{definition}
\newtheorem{definition}[theorem]{Definition}
\theoremstyle{remark}
\newtheorem*{remark}{Remark}

\title{\Large\bf Tangential Surface Pressures and Boundary Localization\\
in Disordered Edwards--Anderson Models}
\author{Hexiang Wang \and Keheng Zhu \and Mauris Chueng}
\date{}

\newcommand{\Addresses}{{%
\bigskip
\footnotesize
\textsc{Hexiang Wang}, School of Mathematical Sciences, Nankai University, Tianjin 300071, China\par
\texttt{Kui6539@outlook.com}
\medskip

\textsc{Keheng Zhu}, Academy for Multidisciplinary Studies, School of Mathematical Sciences, Capital Normal University, Beijing 100048, China\par
\texttt{hexistartop@gmail.com}
\medskip

\textsc{Mauris Chueng}, School of Statistics and Data Science, Jilin University of Finance and Economics,\linebreak Changchun 130117, China\par
\texttt{maurischueng@gmail.com}
}}

\begin{document}
\maketitle

\begin{abstract}
We study boundary free energies in disordered Edwards--Anderson Ising models. For rectangular strips of fixed width, we prove that the expected free-to-fixed boundary correction, divided by the tangential length, converges for every finite inverse temperature and at zero temperature. The proof combines a uniform almost-additivity estimate with product-measure concentration, yielding self-averaging along tangential intervals. We derive an exact finite-volume Gaussian interpolation identity and prove that a uniform exponential boundary-mixing condition implies boundary localization with an explicit exponential rate. The tangential pressure theorem extends to all dimensions and symmetric coupling laws with finite first moment; under a product-concentration hypothesis, self-averaging is obtained along tangential cubes. Finally, a subcritical open-bond criterion, verified explicitly for sparse signed couplings, gives low-temperature localization without assuming an infinite-volume Gibbs-state property.
\end{abstract}
\tableofcontents
\section{Introduction}

Imagine measuring the same disordered sample twice. In the first experiment the spins at one face are allowed to fluctuate; in the second they are coupled to a frozen exterior. The extensive free energy is the same to leading order, but their difference can retain a contribution proportional to the length of the face. In a disordered magnet this contribution is not automatically a domain-wall free energy: it contains local boundary response, and it may have a nonzero quenched mean.

The distinction matters particularly for the Edwards-Anderson model. A free-to-fixed correction, a periodic-to-antiperiodic difference, and a centered sample fluctuation are different random observables. Their normalizations and limiting mechanisms need not agree. The first paper in this line of work gave one-dimensional van Hove counterexamples, a high-temperature surface theorem under a strict Dobrushin condition, and a finite-volume Gaussian interpolation identity \cite{WangZhuChueng2026}. The present paper addresses the low-temperature Gaussian problem by proving a rigorous partial result that does not assume uniqueness of the infinite-volume Gibbs state.

The standard model is defined on $\Z^d$ with independent centered Gaussian couplings. To make the usual cube observable precise in two dimensions, let $\Lambda_N=\{1,\ldots,N\}^2$, let $E_N$ be its internal nearest-neighbor edges, and let
\[
\partial_{\rm e}\Lambda_N=\{\{x,y\}:x\in\Lambda_N,\ y\notin\Lambda_N,\ |x-y|_1=1\},
\qquad |\partial_{\rm e}\Lambda_N|=4N.
\]
With exterior spins fixed to $+1$, define
\[
H_N^\free=-\sum_{\{x,y\}\in E_N}J_{xy}\sigma_x\sigma_y,
\qquad
H_N^\fix=H_N^\free-\sum_{\{x,y\}\in\partial_{\rm e}\Lambda_N}J_{xy}\sigma_x,
\]
and $F_N^b=-\beta^{-1}\log\sum_\sigma e^{-\beta H_N^b}$. Thus $|\partial\Lambda_N|$ below means the crossing-edge count $|\partial_{\rm e}\Lambda_N|$. The natural question is whether
\begin{equation}
\frac{\E(F_N^\fix-F_N^\free)}{|\partial\Lambda_N|}
\end{equation}
has a full-sequence limit at fixed low temperature. For bounded interactions in the Dobrushin uniqueness regime this is a boundary-layer problem. For unbounded Gaussian couplings at low temperature, however, the required boundary localization is not known. The finite-volume identity alone does not provide it.

Our contribution is a reduction with an unconditional tangential theorem and a quantitative conditional localization result. We work in two dimensions and consider a strip of fixed normal thickness $L$ and tangential length $n$, with a fixed random boundary field on the bottom face and free conditions on the remaining faces. We prove:
\begin{enumerate}[label=(\roman*),leftmargin=2.1em]
\item the expected bottom-face correction is almost additive in $n$ up to an error depending only on $L$;
\item the normalized correction converges in expectation, in $L^2$, and almost surely along tangential intervals in a common Gaussian environment;
\item the same conclusions hold for ground-state energy differences at $\beta=\infty$;
\item a boundary-mixing hypothesis implies exponential localization, with an explicit rate, and hence convergence as $L\to\infty$ to a half-space boundary pressure;
\item the tangential pressure construction extends to general $d$ and integrable symmetric coupling laws; under the product-concentration assumption in Theorem~\ref{thm:d-pressure}, self-averaging is proved along tangential cubes, while a sparse-coupling open-bond criterion gives a verifiable low-temperature localization regime.
\end{enumerate}

The strip theorem is deliberately weaker than the unresolved cubic statement, but it has two advantages. It is valid at all temperatures, including zero temperature, and it separates the tangential thermodynamic limit from the normal boundary-state problem. The localization theorem below is conditional: it supplies a checkable route from exponential boundary mixing to the missing normal-direction estimate, while the corresponding mixing property for the ordinary low-temperature Gaussian EA model remains open.

\section{The Gaussian strip model}

Fix an integer $L\geq 1$ and for $n\geq 1$, let
\begin{equation}
\Lambda_{L,n}=\{1,\ldots,L\}\times\{1,\ldots,n\}\subset\Z^2.
\end{equation}
The internal nearest-neighbor edges are denoted by $E_{L,n}$. We single out the bottom boundary edges
\begin{equation}
B_{L,n}^{\mathrm{bot}}
=\{e_j=((1,j),(0,j)):1\leq j\leq n\}.
\end{equation}
The exterior spins at $(0,j)$ are fixed to $+1$. The side and top boundaries are free. This single-face observable is the correct building block for a full surface term; it is not a domain-wall observable.

Let $(J_a)_{a\in E_{L,n}}$ and $(K_j)_{1\leq j\leq n}$ be independent centered Gaussian variables with variance $v>0$. All variables for all $L,n$ are restrictions of one independent Gaussian field on the corresponding infinite strip. For $\sigma\in\{-1,+1\}^{\Lambda_{L,n}}$ put
\begin{equation}
H_{L,n}^{\free}(\sigma)
=-\sum_{a=\{x,y\}\in E_{L,n}}J_a\sigma_x\sigma_y.
\end{equation}
For $r\in[0,1]$, set
\begin{equation}
H_{L,n,r}^{\fix}(\sigma)
=H_{L,n}^{\free}(\sigma)-r\sum_{j=1}^{n}K_j\sigma_{(1,j)}.
\end{equation}
For $\beta\in(0,\infty)$, define
\begin{equation}
F_{L,n}^{\free}=-\frac{1}{\beta}\log Z_{L,n}^{\free},\qquad
F_{L,n,r}^{\fix}=-\frac{1}{\beta}\log Z_{L,n,r}^{\fix}.
\end{equation}
The boundary correction is
\begin{equation}
\Delta F_{L,n}(r)=F_{L,n,r}^{\fix}-F_{L,n}^{\free},\qquad
\Delta F_{L,n}=\Delta F_{L,n}(1).
\end{equation}

At zero temperature replace $F$ by the ground-state energy
\begin{equation}
E_{L,n}^{b}=\min_{\sigma}H_{L,n}^{b}(\sigma),\qquad
\Delta E_{L,n}=E_{L,n}^{\fix}-E_{L,n}^{\free}.
\end{equation}

\begin{lemma}[Pointwise boundary comparison]\label{lem:pointwise}
For every realization of the couplings and every $\beta\in(0,\infty]$,
\begin{equation}
-\sum_{j=1}^{n}|K_j|\leq \Delta F_{L,n}\leq 0,
\end{equation}
with $\Delta F_{L,n}=\Delta E_{L,n}$ at $\beta=\infty$.
\end{lemma}

\begin{proof}
For finite $\beta$, global spin-flip invariance of the free Gibbs measure gives
\begin{equation}
\frac{Z_{L,n}^{\fix}}{Z_{L,n}^{\free}}
=\left\langle \exp\left(\beta\sum_{j=1}^{n}K_j\sigma_{(1,j)}\right)\right\rangle_{\free}
=\left\langle \cosh\left(\beta\sum_{j=1}^{n}K_j\sigma_{(1,j)}\right)\right\rangle_{\free}.
\end{equation}
The ratio is at least one and at most $\exp(\beta\sum_j|K_j|)$, which proves the claim after taking $-\beta^{-1}\log$. At zero temperature, let $\sigma^*$ minimize $H_{L,n}^{\free}$. One of $\sigma^*$ and $-\sigma^*$ has $\sum_jK_j\sigma^*_{(1,j)}\geq0$, so global spin-flip invariance gives $E_{L,n}^{\fix}\leq E_{L,n}^{\free}$. Conversely, $H_{L,n}^{\fix}(\sigma)\geq H_{L,n}^{\free}(\sigma)-\sum_j|K_j|$ for every $\sigma$, which gives the lower bound.
\end{proof}

\begin{figure}[t]
\centering
\includegraphics[width=0.94\textwidth]{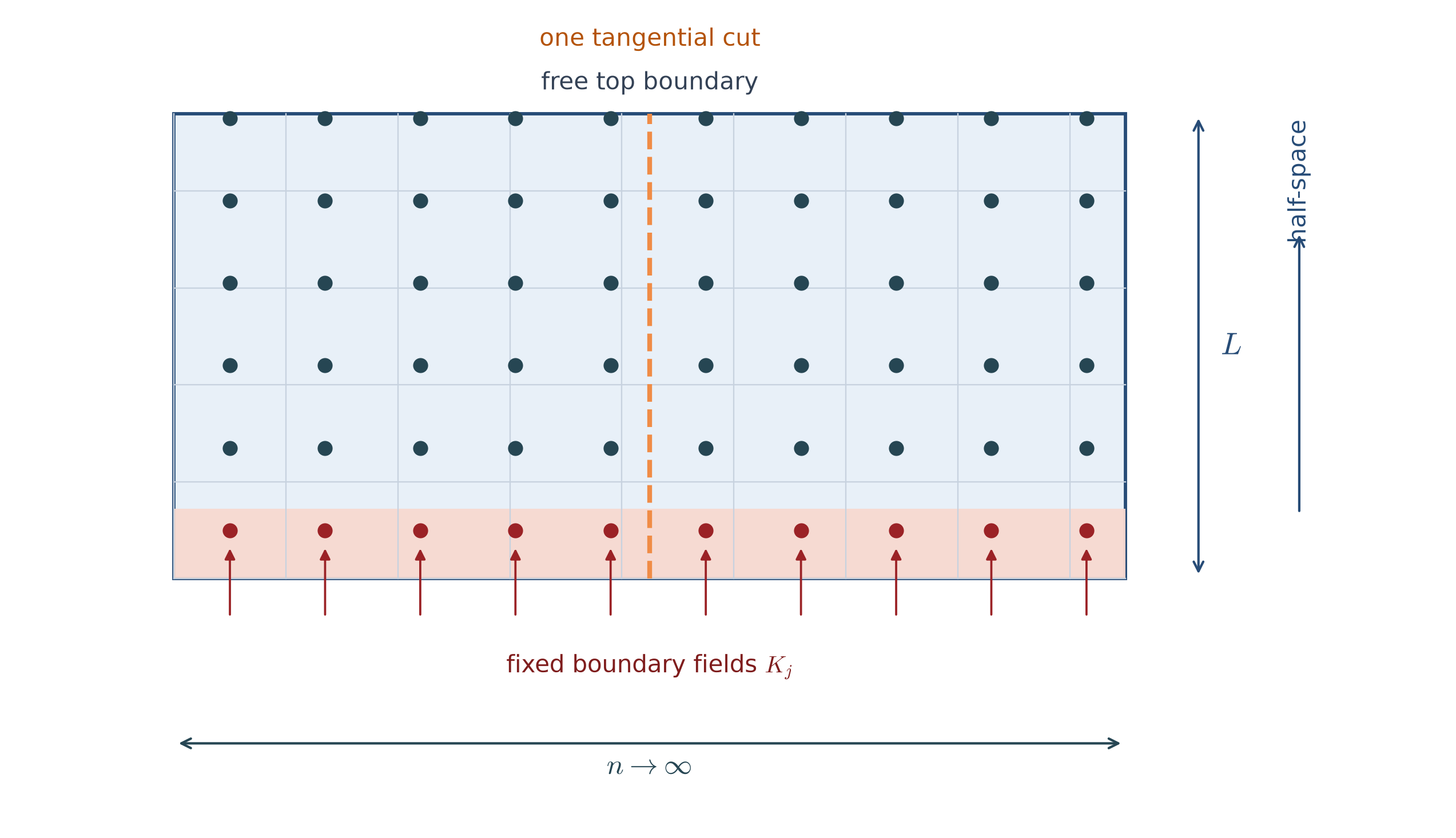}
\caption{The fixed-thickness strip used in the tangential surface pressure. The bottom face carries independent Gaussian boundary fields, the top face is free, and a vertical cut removes $L$ internal bonds. The almost-additivity error is therefore of order $L$, independent of the tangential length $n$. The limit $n\to\infty$ is unconditional; the remaining limit $L\to\infty$ requires boundary localization.}
\label{fig:strip}
\end{figure}

\section{Finite-volume interpolation and concentration}

For $r\in[0,1]$, set
\begin{equation}
m_{L,n,j}(r)=\left\langle \sigma_{(1,j)}\right\rangle_{L,n,r}^{\fix}.
\end{equation}
The next identity is finite-volume and uses no uniqueness assumption.

\begin{proposition}[Gaussian boundary interpolation]\label{prop:interpolation}
For every $L,n$ and $\beta\in(0,\infty)$,
\begin{equation}
\E\Delta F_{L,n}
=-\beta v\int_0^1 r\sum_{j=1}^{n}
\E\left[1-m_{L,n,j}(r)^2\right]\d r.
\label{eq:interpolation}
\end{equation}
\end{proposition}

\begin{proof}
Set $\Phi_{L,n}(r)=\E\log Z_{L,n,r}^{\fix}$. Differentiation at finite volume gives
\begin{equation}
\Phi_{L,n}'(r)
=\beta\sum_{j=1}^{n}\E\left[K_jm_{L,n,j}(r)\right].
\end{equation}
Gaussian integration by parts and the covariance identity
\begin{equation}
\frac{\partial}{\partial K_j}m_{L,n,j}(r)
=\beta r\left(1-m_{L,n,j}(r)^2\right)
\end{equation}
yield
\begin{equation}
\Phi_{L,n}'(r)
=\beta^2 v r\sum_{j=1}^{n}
\E\left[1-m_{L,n,j}(r)^2\right].
\end{equation}
Since $\E\Delta F_{L,n}=-\beta^{-1}(\Phi_{L,n}(1)-\Phi_{L,n}(0))$, integration over $r$ proves \cref{eq:interpolation}.
\end{proof}

\begin{proposition}[Gaussian concentration]\label{prop:concentration}
For finite $\beta$ and every $u>0$,
\begin{equation}
\Pp\left(\left|\Delta F_{L,n}-\E\Delta F_{L,n}\right|\geq u\right)
\leq 2\exp\left(
 -\frac{u^2}{2v\left(4[L(n-1)+(L-1)n]+n\right)}
 \right).
\label{eq:concentration}
\end{equation}
The same bound holds for $\Delta E_{L,n}$ at zero temperature.
\end{proposition}

\begin{proof}
At differentiability points,
\begin{equation}
\left|\frac{\partial\Delta F_{L,n}}{\partial J_a}\right|\leq 2,\qquad
\left|\frac{\partial\Delta F_{L,n}}{\partial K_j}\right|\leq 1.
\end{equation}
There are $L(n-1)+(L-1)n$ internal bonds and $n$ boundary bonds. Thus the squared Euclidean Lipschitz constant after rescaling by the Gaussian variance is at most $v\{4[L(n-1)+(L-1)n]+n\}$. The Gaussian concentration inequality gives \cref{eq:concentration}. The ground-state energy is a Lipschitz function of every coupling with the same coordinate bounds, so the zero-temperature assertion follows by the same concentration theorem.
\end{proof}

\section{Tangential surface pressure}

Cutting the strip between columns $m$ and $m+1$ removes precisely $L$ internal bonds. The two resulting components have the same law as independent strips of lengths $m$ and $n$.

\begin{lemma}[Almost-additivity]\label{lem:almostadd}
Let $a_{L,n}=\E\Delta F_{L,n}$ at finite temperature, or $a_{L,n}=\E\Delta E_{L,n}$ at zero temperature. Then
\begin{equation}
\left|a_{L,m+n}-a_{L,m}-a_{L,n}\right|
\leq C_L,\qquad
C_L=2L\sqrt{\frac{2v}{\pi}}.
\label{eq:almostadd}
\end{equation}
\end{lemma}

\begin{proof}
Let $\Delta F_{L,m+n}^{\mathrm{cut}}$ denote the same free-energy difference after setting the $L$ couplings crossing the cut equal to zero. The Hamiltonian then factorizes, and independence of the disorder in the two components gives
\begin{equation}
\E\Delta F_{L,m+n}^{\mathrm{cut}}=a_{L,m}+a_{L,n}.
\end{equation}
Changing one internal coupling from $J$ to $0$ changes each of the two free energies by at most $|J|$. Hence
\begin{equation}
\left|\Delta F_{L,m+n}-\Delta F_{L,m+n}^{\mathrm{cut}}\right|
\leq 2\sum_{q=1}^{L}|J_q|.
\end{equation}
Taking expectations proves \cref{eq:almostadd}. The same one-bond Lipschitz estimate holds for ground-state energies.
\end{proof}

\begin{theorem}[Existence of the fixed-thickness surface pressure]\label{thm:strip}
For every fixed $L\geq 1$ and $\beta\in(0,\infty]$, the limit
\begin{equation}
s_L(\beta)=\lim_{n\to\infty}\frac{\E\Delta F_{L,n}}{n}
\end{equation}
exists, where $\Delta F_{L,n}$ is interpreted as $\Delta E_{L,n}$ when $\beta=\infty$. Moreover,
\begin{equation}
-\sqrt{\frac{2v}{\pi}}\leq s_L(\beta)\leq 0.
\label{eq:pressurebound}
\end{equation}
For a common realization of the Gaussian field,
\begin{equation}
\frac{\Delta F_{L,n}}{n}\longrightarrow s_L(\beta)
\quad\text{almost surely and in }L^2,
\label{eq:stripas}
\end{equation}
and the same assertion holds with $\Delta E_{L,n}$ at zero temperature.
\end{theorem}

\begin{proof}
By \cref{lem:almostadd}, $b_n=a_{L,n}+C_L$ is subadditive. Fekete's lemma gives the existence of $\lim_n b_n/n$, hence of $\lim_n a_{L,n}/n$. The pointwise comparison in \cref{lem:pointwise} and the law of large numbers give \cref{eq:pressurebound}.
 
For fixed $L$, apply \cref{prop:concentration} with $u=\varepsilon n$. Since
\[
4[L(n-1)+(L-1)n]+n\leq (8L+1)n,
\]
the resulting upper bound is summable in $n$. Borel--Cantelli yields almost-sure convergence of $\Delta F_{L,n}/n$ to its deterministic mean limit. Integrating the same Gaussian tail gives convergence in $L^2$. The zero-temperature statement follows from the second part of \cref{prop:concentration}.
\end{proof}

\begin{remark}
Theorem \ref{thm:strip} is a thermodynamic limit in the tangential direction. It does not assert that $s_L(\beta)$ converges as $L\to\infty$. That normal-direction problem is precisely where low-temperature boundary states enter.
\end{remark}

\section{The boundary-localization criterion}

The strip pressure can be related to a half-space response without assuming Dobrushin uniqueness, provided one supplies a quantitative localization estimate. We formulate the estimate in a way that can be checked by coupling, metastate, or numerical boundary-response methods.

Let $\mathcal S_L=\{1,\ldots,L\}\times\mathbb Z$ denote the infinite strip with the bottom fields and with no interactions across its top. Let $p_L(r)$ denote the limiting expected pressure difference per tangential length obtained by the almost-additivity argument, with boundary-field amplitude $r$. The function $p_L$ is concave and is Lipschitz in $r$ (with constant $\E|K_0|$), hence absolutely continuous. At almost every $r>0$, define
\begin{equation}
u_L(r)=-\frac{p_L'(r)}{\beta v r}.
\end{equation}
We call a family of tangentially translation-covariant DLR states $(\mu_{L,r},\mu_{\Hh,r})_{r\in[0,1]}$ response-compatible when it is jointly measurable and, for almost every $r>0$,
\begin{equation}
u_L(r)=\E\left[1-\left\langle\sigma_{(1,0)}\right\rangle_{\mu_{L,r}}^2\right].
\label{eq:strip-response-state}
\end{equation}
For such a family set
\begin{equation}
u_{\Hh}(r)=\E\left[1-\left\langle\sigma_{(1,0)}\right\rangle_{\mu_{\Hh,r}}^2\right].
\end{equation}
We require $r\mapsto u_{\Hh}(r)$ to be measurable. The half-space is $\Hh=\{(x_1,x_2)\in\mathbb Z^2:x_1\geq1\}$ with boundary field $rK$ on $x_1=0$.
At exceptional values where $p_L'(r)$ is not defined (including $r=0$), we take $u_L(r)$ to be given by the displayed strip state-response expectation; this makes both response functions defined on $[0,1]$ without affecting any integral below.

\begin{definition}[Boundary localization]
We say that a response-compatible family satisfies the boundary-localization condition if there is a deterministic function $\rho(L)\downarrow 0$ such that
\begin{equation}
\sup_{0\leq r\leq 1}
\E\left|
\left\langle\sigma_{(1,0)}\right\rangle_{\mu_{L,r}}
-\left\langle\sigma_{(1,0)}\right\rangle_{\mu_{\Hh,r}}
\right|
\leq \rho(L).
\label{eq:BL}
\end{equation}
\end{definition}

We now give a sufficient condition which produces an explicit choice of $\rho$. For a finite region $V\subset\Hh$ and a boundary condition $\eta$, write $\langle\,\cdot\,\rangle_{V,r}^{\eta}$ for the Gibbs expectation with the bottom field $rK$ kept fixed and with $\eta$ prescribing the remaining external spins. The distance below is the nearest-neighbor distance in $\Z^2$.

\begin{definition}[Uniform exponential boundary mixing]\label{def:EBM}
Fix $\beta\in(0,\infty)$. We say that the environment satisfies uniform exponential boundary mixing at inverse temperature $\beta$ if there are $m_\beta>0$ and a nonnegative random variable $A_\beta$ with $\E A_\beta<\infty$ such that, almost surely, for every finite $V\subset\Hh$, every $r\in[0,1]$, every site $x\in V$, and every pair of admissible boundary conditions $\eta,\eta'$ which differ only on $D\subset\partial V$, one has
\begin{equation}
\left|\left\langle\sigma_x\right\rangle_{V,r}^{\eta}
-\left\langle\sigma_x\right\rangle_{V,r}^{\eta'}\right|
\leq A_\beta\sum_{y\in D}\exp\left(-m_\beta\,\dist(x,y)\right).
\label{eq:EBM}
\end{equation}
The constants are uniform in $V$, $r$, and the boundary conditions. Here an admissible boundary law is any probability mixture of fixed external-spin configurations, with arbitrary (possibly disorder-dependent) mixing weights; the free boundary law obtained by deleting the corresponding crossing interactions is included explicitly among the admissible laws. The same inequality is required after averaging over any two such laws whose supports differ only on $D$.
\end{definition}

This is a boundary-to-bulk mixing condition rather than a full one-site Dobrushin condition \cite{Georgii2011,Ruelle1969}. In particular, it only controls the influence of the remote boundary on the observable at $x$. The integrable random prefactor allows for the unbounded Gaussian couplings; a deterministic prefactor is the special case $A_\beta\equiv C_\beta$. The condition is a genuine additional hypothesis in the low-temperature Gaussian EA model, not a consequence of the finite-volume interpolation identity.

For a low-temperature regime one fixes $\beta_0>0$ and asks for \cref{eq:EBM} at every finite $\beta\geq\beta_0$, with constants allowed to depend on $\beta$. If $m_\beta$ and $\E A_\beta$ can be chosen uniformly on a compact inverse-temperature interval, then the localization rate in \cref{eq:rho-explicit} is uniform on that interval as well.

\begin{theorem}[Surface limit under boundary localization]\label{thm:localization}
Assume the boundary-localization condition \cref{eq:BL} for a response-compatible jointly measurable family. Then, for every finite $\beta$,
\begin{equation}
\lim_{L\to\infty}s_L(\beta)
= -\beta v\int_0^1 r\,u_{\Hh}(r)\d r.
\label{eq:halfspacepressure}
\end{equation}
The limit is independent of the selected half-space state whenever the right-hand side is selection-independent.
\end{theorem}

\begin{proof}
Absolute continuity of $p_L$, the response-compatibility identity, and $p_L(0)=0$ give
\begin{equation}
s_L(\beta)=-\beta v\int_0^1 r\,u_L(r)\d r.
\label{eq:stripresponse}
\end{equation}
This identity is an equality of integrals; the derivative statement is needed only almost everywhere. Since $|x^2-y^2|\leq2|x-y|$ on $[-1,1]$, \cref{eq:BL} implies
\begin{equation}
|u_L(r)-u_{\Hh}(r)|\leq2\rho(L)\quad\text{for almost every }r\in[0,1].
\end{equation}
The integrands are bounded by one, and dominated convergence in \cref{eq:stripresponse} proves \cref{eq:halfspacepressure}.
\end{proof}

\begin{lemma}[Top-boundary summation]
For $m>0$ and $L\geq1$, let
\begin{equation}
T_L=\{(L+1,j):j\in\Z\}.
\end{equation}
Then, with $x_0=(1,0)$,
\begin{equation}
\sum_{y\in T_L}\exp\left(-m\,\dist(x_0,y)\right)
=c(m)\exp(-mL),
\qquad
c(m)=\frac{1+\exp(-m)}{1-\exp(-m)}.
\label{eq:top-sum}
\end{equation}
\end{lemma}

\begin{proof}
For $y=(L+1,j)$ one has $\dist(x_0,y)=L+|j|$. Summing the resulting geometric series gives
\begin{equation}
\sum_{j\in\Z}\exp\left(-m(L+|j|)\right)
=\exp(-mL)\left(1+2\sum_{j\geq1}\exp(-mj)\right),
\end{equation}
which is \cref{eq:top-sum}.
\end{proof}

\begin{theorem}[Exponential localization from boundary mixing]\label{thm:exp-localization}
Fix $\beta\in(0,\infty)$ and suppose that uniform exponential boundary mixing \cref{eq:EBM} holds. Then the strip and half-space specifications have unique DLR states, denoted by $\mu_{L,r}$ and $\mu_{\Hh,r}$, and these states form a jointly measurable response-compatible family. Moreover,
\begin{equation}
\sup_{0\leq r\leq1}
\E\left|
\left\langle\sigma_{(1,0)}\right\rangle_{\mu_{L,r}}
-\left\langle\sigma_{(1,0)}\right\rangle_{\mu_{\Hh,r}}
\right|
\leq \rho_\beta(L),
\qquad
\rho_\beta(L)=c(m_\beta)\,\E A_\beta\,\exp(-m_\beta L).
\label{eq:rho-explicit}
\end{equation}
In particular, the boundary-localization condition holds with an exponential rate. Moreover, if
\begin{equation}
\tau_\beta^{\Hh}=-\beta v\int_0^1r\,u_{\Hh}(r)\d r,
\end{equation}
then
\begin{equation}
\left|s_L(\beta)-\tau_\beta^{\Hh}\right|
\leq \beta v\,c(m_\beta)\,\E A_\beta\,\exp(-m_\beta L).
\label{eq:surface-rate}
\end{equation}
\end{theorem}

\begin{proof}
We first record why the state and response assertions are part of the conclusion. Exhaust either $\mathcal S_L$ or $\Hh$ by deterministic finite cylinders. For a fixed site $x$, the sum in \cref{eq:EBM} over the boundary of a cylinder whose tangential cutoff is $R$ is bounded by a polynomial in $R$ times $e^{-m_\beta R}$ (the bottom boundary is kept fixed). Hence two subsequential DLR limits have identical one-site expectations at every $x$. Applying the same estimate successively to the sites of any fixed finite set, with already exposed spins absorbed into the boundary law, gives equality of all cylinder marginals and therefore equality of the states. Compactness gives existence. Finite-volume expectations are measurable in $(r,\omega)$, so their unique pointwise limits give joint measurability; translating the exhaustion gives tangential covariance.

For the strip, the finite-volume Gaussian identity gives
\[
p_{L,n}'(r)=-\beta v r\,\frac1n\sum_{j=1}^n
\E\left[1-m_{L,n,j}(r)^2\right].
\]
Applying \cref{eq:EBM} to a site at tangential distance $h$ from either lateral cutoff bounds the difference between its finite-cylinder magnetization and the unique strip-state magnetization by $A_\beta C_L(e^{-m_\beta h})$, where $C_L$ is a finite geometric sum over the $L$ layers. The average of these errors is $O(1/n)$ after expectation. Convergence of concave functions and their derivatives at every differentiability point of $p_L$ therefore gives the response-compatibility identity \cref{eq:strip-response-state} for almost every $r>0$.

Now set $C_R=\{1,\ldots,L\}\times\{-R,\ldots,R\}$. Disintegrate the half-space state over spins outside $C_R$; its conditional law induces a (possibly random) admissible boundary law on the top and lateral sides of $C_R$. Disintegrate the strip state similarly. The two conditional laws agree on the bottom field. Their boundary laws can differ on the top set
\[
T_{L,R}=\{(L+1,j):-R\leq j\leq R\}
\]
and on the two lateral sets. The mixing hypothesis, including its mixture clause, gives for $x_0=(1,0)$
\begin{equation}
\left|
\left\langle\sigma_{x_0}\right\rangle_{\mu_{L,r}}
-\left\langle\sigma_{x_0}\right\rangle_{\mu_{\Hh,r}}
\right|
\leq A_\beta\left(\sum_{y\in T_{L,R}}e^{-m_\beta\dist(x_0,y)}+2L e^{-m_\beta(R+1)}\right).
\end{equation}
The second term explicitly bounds the two lateral faces, each having $L$ sites at distance at least $R+1$. Letting $R\to\infty$ and using \cref{eq:top-sum}, then taking expectation, proves \cref{eq:rho-explicit}, uniformly in $r$.

Since $|a^2-b^2|\leq2|a-b|$ for $a,b\in[-1,1]$, \cref{eq:rho-explicit} yields
\begin{equation}
|u_L(r)-u_{\Hh}(r)|
\leq2\rho_\beta(L)\quad\text{for almost every }r\in[0,1].
\end{equation}
Combining this estimate with the strip response identity in the proof of \cref{thm:localization} gives
\begin{equation}
\left|s_L(\beta)-\tau_\beta^{\Hh}\right|
\leq\beta v\int_0^1 2r\rho_\beta(L)\d r
=\beta v\rho_\beta(L),
\end{equation}
which is \cref{eq:surface-rate}.
\end{proof}

No full-cube surface formula is claimed here. The single-face estimate above does not by itself decompose the simultaneous free-to-fixed correction into four face terms: boundary conditions on one face can influence an entire other face. A valid four-face theorem would require a uniform mixing hypothesis with arbitrary side boundary laws and a telescoping decomposition whose non-face remainder is $O(1)$ in expectation. Under that additional hypothesis, the crossing-edge normalization defined in the introduction would give a coefficient $1/4$ in front of the sum of the four oriented face pressures, since $|\partial_{\rm e}\Lambda_N|=4N$.

\begin{remark}
Theorem \ref{thm:localization} is a criterion, not a proof of \cref{eq:BL} for the low-temperature Gaussian EA model. Establishing \cref{eq:BL}, or constructing a sequence for which it fails, remains the central normal-direction problem.
\end{remark}

\section{General dimensions and coupling laws}

The tangential argument is not specific to two dimensions or to Gaussian disorder. We record the corresponding formulation, separating the assumptions needed for the pressure limit from those used only for Gaussian concentration and interpolation.

\subsection{Tangential pressure and conditional localization}

Let $d\geq2$, let $k=d-1$, and write
\begin{equation}
\Lambda_{L,\mathbf n}^{(d)}
=\{1,\ldots,L\}\times\prod_{i=2}^{d}\{1,\ldots,n_i\},
\qquad
|\mathbf n|=\prod_{i=2}^{d}n_i.
\label{eq:d-slab}
\end{equation}
The bottom face has $|\mathbf n|$ independent boundary fields. Let the internal couplings be i.i.d. with a symmetric law $\nu$ and let the boundary fields have a symmetric law $\nu_{\mathrm{b}}$, independent of the internal disorder. Set
\begin{equation}
\mu_{\mathrm{int}}=\int_{\mathbb R}|x|\,\nu(\d x),
\qquad
\mu_{\mathrm{b}}=\int_{\mathbb R}|x|\,\nu_{\mathrm{b}}(\d x).
\end{equation}
For the pressure limit we require only $\mu_{\mathrm{int}}+\mu_{\mathrm{b}}<\infty$. The notation $F_{L,\mathbf n,r}^{\nu}$ and $\Delta F_{L,\mathbf n}^{\nu}$ refers to the same free-to-fixed construction as before, with the Gaussian law replaced by $\nu$ and $\nu_{\mathrm b}$. For vectors $\mathbf n,\mathbf m$ with $n_j=m_j$ for $j\ne i$, the concatenation $\mathbf n\oplus_i\mathbf m$ is the vector with $i$th coordinate $n_i+m_i$ and all other coordinates $n_j$. Products indexed by $j\ne i$ below mean $j\in\{2,\ldots,d\}\setminus\{i\}$.

\begin{lemma}[Multidimensional tangential almost-additivity]
Let $a_{L,\mathbf n}^{\nu}=\E\Delta F_{L,\mathbf n}^{\nu}$ at finite temperature, or let it denote the corresponding expected ground-state correction at zero temperature. If two tangential boxes are concatenated in direction $i\in\{2,\ldots,d\}$, then
\begin{equation}
\left|a_{L,\mathbf n\oplus_i\mathbf m}^{\nu}
-a_{L,\mathbf n}^{\nu}-a_{L,\mathbf m}^{\nu}\right|
\leq 2L\mu_{\mathrm{int}}\prod_{j\in\{2,\ldots,d\}\setminus\{i\}}n_j,
\label{eq:d-almostadd}
\end{equation}
where the two rectangles have the same side lengths in the directions transverse to $i$.
\end{lemma}

\begin{proof}
Cutting in direction $i$ removes $L\prod_{j\in\{2,\ldots,d\}\setminus\{i\}}n_j$ internal bonds. Setting these couplings to zero factorizes the Hamiltonian. Changing one coupling changes each of the two free energies by at most its absolute value, and the same estimate holds for ground-state energies. Taking expectations gives \cref{eq:d-almostadd}.
\end{proof}

\begin{theorem}[Fixed-thickness pressure in general dimension]\label{thm:d-pressure}
For every $d\geq2$, every fixed $L\geq1$, and every $\beta\in(0,\infty]$, the van Hove limit
\begin{equation}
s_{L,\nu}^{(d)}(\beta)
=\lim_{\min_i n_i\to\infty}
\frac{\E\Delta F_{L,\mathbf n}^{\nu}}{|\mathbf n|}
\label{eq:d-pressure}
\end{equation}
exists. At zero temperature, $\Delta F$ is interpreted as the ground-state correction. Moreover,
\begin{equation}
-\mu_{\mathrm b}\leq s_{L,\nu}^{(d)}(\beta)\leq0.
\label{eq:d-pressure-bound}
\end{equation}
If the product law of all internal and boundary variables satisfies, for some finite $\kappa>0$, the concentration property
\begin{equation}
\Pp\left(|f-\E f|\geq t\right)
\leq2\exp\left(-\frac{t^2}{\kappa\sum_a c_a^2}\right)
\label{eq:product-concentration}
\end{equation}
for every coordinatewise Lipschitz function with coordinate constants $c_a$ \cite{Ledoux2001,BoucheronLugosiMassart2013}, then the normalized corrections converge almost surely and in $L^2$ along tangential cubes $n_2=\cdots=n_d\to\infty$.
\end{theorem}

\begin{proof}
The pointwise comparison \cref{lem:pointwise} uses only global spin-flip symmetry and gives
\begin{equation}
-\sum_{x\text{ on the bottom face}}|K_x|
\leq\Delta F_{L,\mathbf n}^{\nu}\leq0.
\end{equation}
Fix $q\geq1$ and let $Q_{\mathbf n}=\prod_{i=2}^d\{1,\ldots,q\lfloor n_i/q\rfloor\}$ be the tiled core. It consists of
\(M_{\mathbf n}=\prod_{i=2}^d\lfloor n_i/q\rfloor\) disjoint translates of the $q^k$ tangential box. Cut every internal seam between these boxes and every bond joining the core to its remainder $R_{\mathbf n}$. The number of seam bonds inside the core is at most
\[
L\sum_{i=2}^d\left(\lfloor n_i/q\rfloor-1\right)\prod_{j\in\{2,\ldots,d\}\setminus\{i\}}q\lfloor n_j/q\rfloor
\leq \frac{kL}{q}|\mathbf n|,
\]
and the core--remainder interface has at most $L\sum_{i=2}^d\prod_{j\in\{2,\ldots,d\}\setminus\{i\}}n_j=o(|\mathbf n|)$ bonds along van Hove boxes. The one-bond comparison used in Lemma~\ref{lem:almostadd} therefore gives
\[
\left|a_{L,\mathbf n}^{\nu}-M_{\mathbf n}a_{L,\mathbf q}^{\nu}\right|
\leq 2\mu_{\mathrm{int}}\left(\frac{kL}{q}|\mathbf n|+o(|\mathbf n|)\right)+\mu_{\mathrm b}|R_{\mathbf n}|.
\]
Here the last term is the pointwise boundary comparison on the (possibly disconnected) remainder, and $|R_{\mathbf n}|/|\mathbf n|\to0$ for fixed $q$. Since $M_{\mathbf n}q^k/|\mathbf n|\to1$, taking limsup and liminf yields an interval of length at most $4kL\mu_{\mathrm{int}}/q$. Letting $q\to\infty$ proves existence of the limit. Dividing the pointwise comparison by $|\mathbf n|$ and applying the law of large numbers proves \cref{eq:d-pressure-bound}.

For the last assertion, the coordinate Lipschitz constants of the free-to-fixed correction are at most $2$ for internal couplings and $1$ for boundary couplings. Along tangential cubes, the number of coordinates is $O_{d,L}(|\mathbf n|)$, so \cref{eq:product-concentration} with $t=\varepsilon|\mathbf n|$ gives a summable tail bound. Borel--Cantelli and integration of the tail yield almost-sure and $L^2$ convergence along those cubes.
\end{proof}

For a lattice direction $\xi\in\{\pm e_1,\ldots,\pm e_d\}$, write $\mathbb H_d^\xi=\{x\in\mathbb Z^d:x\cdot\xi\geq1\}$ and obtain its specification by the corresponding coordinate rotation; the notation below omits $\xi$ and uses $\xi=e_1$. Let $p_{L,\nu}^{(d)}(r)=\lim_{\min_i n_i\to\infty}\E\Delta F_{L,\mathbf n,r}^{\nu}/|\mathbf n|$, whose existence follows from the preceding block argument. We call a jointly measurable tangentially translation-covariant strip family response-compatible when
\[
\bigl(p_{L,\nu}^{(d)}\bigr)'(r)
=-\E\left[K_0\left\langle\sigma_{(1,0,\ldots,0)}\right\rangle_{\mu_{L,r}}\right]
\]
for almost every $r\in(0,1)$. For such a family, absolute continuity of the pressure gives the elementary boundary-response formula
\begin{equation}
s_{L,\nu}^{(d)}(\beta)
=-\int_0^1\E\left[K_0\left\langle\sigma_{(1,0,\ldots,0)}\right\rangle_{\mu_{L,r}}\right]\d r,
\label{eq:general-response}
\end{equation}
since $p_{L,\nu}^{(d)}(0)=0$. The Gaussian identity \cref{eq:interpolation} is recovered from \cref{eq:general-response} by Gaussian integration by parts.

\begin{proposition}[General-dimensional exponential localization]\label{prop:d-localization}
Let $\mathbb H_d=\{(x_1,z)\in\mathbb Z\times\mathbb Z^{d-1}:x_1\geq1\}$. Assume the boundary-mixing condition defined above, with finite $V\subset\mathbb H_d$, nearest-neighbor distance in $\mathbb Z^d$, and all admissible boundary laws. Suppose also that
\[
M_{\beta,K}:=\sup_{z\in\mathbb Z^{d-1}}\E[|K_z|A_\beta]<\infty.
\]
Then the slab and half-space DLR states are unique, jointly measurable, and response-compatible, and
\begin{equation}
\sup_{0\leq r\leq1}
\E\left|
\left\langle\sigma_{(1,0,\ldots,0)}\right\rangle_{\mu_{L,r}}
-\left\langle\sigma_{(1,0,\ldots,0)}\right\rangle_{\mu_{\mathbb H_d,r}}
\right|
\leq \rho_{\beta,d}(L),
\label{eq:d-rho}
\end{equation}
where
\begin{equation}
\rho_{\beta,d}(L)
=c_{d-1}(m_\beta)\,\E A_\beta\,\exp(-m_\beta L),
\qquad
c_{d-1}(m)=\left(\frac{1+\exp(-m)}{1-\exp(-m)}\right)^{d-1}.
\label{eq:d-rho-explicit}
\end{equation}
Moreover, with
\begin{equation}
\tau_{\beta,\nu}^{\mathbb H_d}
=-\int_0^1\E\left[K_0\left\langle\sigma_{(1,0,\ldots,0)}\right\rangle_{\mu_{\mathbb H_d,r}}\right]\d r,
\end{equation}
For a face with normal $\xi$, the oriented half-space pressure $\tau_{\beta,\nu}^{\mathbb H_d}(\xi)$ is defined by the same formula after the rotation specified above; the displayed quantity is the value at $\xi=e_1$. The surface-rate estimate is
\begin{equation}
\left|s_{L,\nu}^{(d)}(\beta)-\tau_{\beta,\nu}^{\mathbb H_d}\right|
\leq c_{d-1}(m_\beta)\,\E[|K_0|A_\beta]\,\exp(-m_\beta L).
\label{eq:d-surface-rate}
\end{equation}
\end{proposition}

\begin{proof}
The exhaustion argument in the proof of \cref{thm:exp-localization} applies with tangential cylinders $[-R,R]^{d-1}$: the lateral-boundary sum is a polynomial of degree $d-2$ in $R$ times $e^{-m_\beta R}$ and tends to zero. It proves uniqueness and joint measurability. Differentiating the finite-volume pressure gives the finite-volume version of the response formula. The same mixing estimate bounds the tangential average of the $K$-weighted lateral-cut error by $M_{\beta,K}$ times a boundary-to-volume ratio; it tends to zero. Convergence of derivatives of concave pressure functions at differentiability points therefore proves response compatibility and hence \cref{eq:general-response}.

The top boundary of a width-$L$ slab is $T_L=\{(L+1,z):z\in\mathbb Z^{d-1}\}$. For $x_0=(1,0,\ldots,0)$,
\begin{equation}
\sum_{z\in\mathbb Z^{d-1}}
\exp\left(-m\,\dist(x_0,(L+1,z))\right)
\begin{aligned}
&=\exp(-mL)\prod_{j=1}^{d-1}\left(1+2\sum_{q\geq1}\exp(-mq)\right)\\
&=c_{d-1}(m)\exp(-mL).
\end{aligned}
\end{equation}
Disintegrate the half-space and slab states over a common finite tangential cylinder, exactly as in the proof of \cref{thm:exp-localization}. The mixture version of boundary mixing applies to the induced laws; the lateral term vanishes as the cutoff tends to infinity. The displayed top-boundary sum then proves \cref{eq:d-rho}. Multiplying the pathwise top-boundary estimate by $|K_0|$, taking expectations, and using \cref{eq:general-response} for the slab and half-space states gives \cref{eq:d-surface-rate}.
\end{proof}

\subsection{A verifiable low-temperature criterion}

The preceding boundary-mixing assumption is not needed for the following sparse model: disagreement percolation on its open-edge graph gives a direct, checkable localization estimate without referring to an infinite-volume Gibbs-state property.

\begin{lemma}[Single-edge influence]\label{lem:single-edge-influence}
For an Ising specification at inverse temperature $\beta$, changing one neighboring spin across an edge with coupling $J$ changes the conditional law at the other endpoint by at most $\tanh(\beta|J|)$.
\end{lemma}

\begin{proof}
The conditional mean at a site with local field $h$ is $\tanh(\beta h)$. Flipping a neighboring spin changes $h$ from $h+J$ to $h-J$, and the maximal total variation difference of the two Bernoulli laws is bounded by $\tanh(\beta|J|)$.
\end{proof}

\begin{theorem}[Sparse open-bond localization]\label{thm:dp-localization}
Let $d\geq2$ and let the internal couplings be i.i.d. with law
\[
\nu_\theta=(1-\theta)\delta_0+\frac{\theta}{2}(\delta_{J_0}+\delta_{-J_0}),
\qquad J_0>0,
\]
where $0\leq\theta<(2d-1)^{-1}$. Put $q_d=(2d-1)\theta$. If the boundary fields are independent of the internal couplings and have a symmetric law with finite first moment, then the strip and half-space states are unique, jointly measurable, and response-compatible for every finite $\beta$ and every $r\in[0,1]$, and, uniformly in $r$,
\begin{equation}
\sup_{0\leq r\leq1}
\E\left|
\left\langle\sigma_{(1,0,\ldots,0)}\right\rangle_{\mu_{L,r}}
-\left\langle\sigma_{(1,0,\ldots,0)}\right\rangle_{\mu_{\mathbb H_d,r}}
\right|
\leq \frac{4d \theta}{1-q_d}\,q_d^{L-1},
\label{eq:dp-rho}
\end{equation}
\begin{equation}
\left|s_{L,\nu_\theta}^{(d)}(\beta)-\tau_{\beta,\nu_\theta}^{\mathbb H_d}\right|
\leq \mu_{\mathrm b}\frac{4d \theta}{1-q_d}\,q_d^{L-1}.
\label{eq:dp-surface-rate}
\end{equation}
\end{theorem}

\begin{proof}
Call an internal edge open when its coupling is nonzero. Conditional on the open-edge graph, the Hamiltonian factorizes over its connected components. Consequently, for any two boundary laws that agree below the top of the slab, a component containing $x_0$ can have different spin marginals only if it contains an open path from $x_0$ to the top boundary. The same statement holds after coupling arbitrary mixtures of boundary conditions. A path of length $n$ from a fixed site is self-avoiding after erasing loops, and there are at most $2d(2d-1)^{n-1}$ such paths. Since each prescribed edge is open with probability $\theta$, a union bound gives
\begin{equation}
\begin{aligned}
\E\left|\left\langle\sigma_{x_0}\right\rangle_{\mu_{L,r}}
-\left\langle\sigma_{x_0}\right\rangle_{\mu_{\mathbb H_d,r}}\right|
&\leq 2\sum_{n\geq L}2d(2d-1)^{n-1}\theta^n\\
&=\frac{4d\theta}{1-q_d}q_d^{L-1}.
\end{aligned}
\label{eq:dp-path-sum}
\end{equation}
The same path bound for the boundary of a graph-distance ball shows that the probability that a fixed site is connected to distance $R$ tends to zero. Thus every open component is finite almost surely (first for a fixed site and then for all sites by countability), so the finite-component specifications define the unique strip and half-space DLR states. Their local expectations are measurable in $(r,\omega)$. At a tangential site a distance $h$ from a finite-volume lateral boundary, the expected $K$-weighted discrepancy from the infinite-volume marginal is at most a constant times $\mu_{\mathrm b}q_d^h$. Averaging this summable bound over a tangential box, then using convergence of derivatives of the concave pressures, proves response compatibility. The estimate is uniform in the bottom field and hence in $r$. Multiplying by $|K_0|$, using independence of $K_0$ from the open graph, and applying \cref{eq:general-response} proves \cref{eq:dp-surface-rate}.
\end{proof}

\begin{corollary}[Sparse signed couplings]
For the law in \cref{thm:dp-localization},
\begin{equation}
q_d=(2d-1)\theta<1.
\end{equation}
The exponential localization therefore holds uniformly for all finite $\beta$, including arbitrarily low temperatures. If $\theta>0$, its exponential rate $-\log q_d$ has a positive, finite limit as $\beta\to\infty$; if $\theta=0$, the localization error is identically zero and the bound is interpreted with $q_d^0=1$.
\end{corollary}

\begin{remark}
The sparse criterion is genuinely checkable from the one-edge law: it is the subcriticality inequality $(2d-1)\theta<1$ for the open-bond graph. This proof uses only connectivity of nonzero couplings and therefore applies at every finite temperature, while making no localization claim for the ordinary full Gaussian EA model.
\end{remark}

\section{Zero temperature and the status of the open problem}

The almost-additivity proof does not use differentiability of the free energy and therefore survives at $\beta=\infty$. This gives a genuine zero-temperature tangential surface pressure for each fixed thickness. The interpolation formula, by contrast, is a finite-temperature identity and should not be passed to $\beta=\infty$ without additional uniform estimates. In particular, a zero-temperature ground-state limit cannot be inferred merely by taking $\beta\to\infty$ in \cref{eq:interpolation}.

Theorem \ref{thm:d-pressure} shows that the fixed-thickness construction is not tied to Gaussian couplings or to two dimensions. The disagreement-percolation theorem gives a separate low-temperature localization result for sparse signed laws, including a regime in which the rate remains positive as $\beta\to\infty$. The ordinary full Gaussian law falls outside this verifiable criterion at low temperature, so its normal-direction problem remains distinct.

The fixed-thickness theorem also clarifies the relation between two common observables. Let $D_N$ be a periodic-to-antiperiodic seam-flip free energy in a square. Symmetry of the seam disorder forces $\E D_N=0$ at each $N$, whereas the free-to-fixed correction has the nonzero mean described by \cref{eq:interpolation}. A variance bound for $D_N$ is therefore not a variance bound for $\Delta F_N$. This distinction is essential in discussions of stiffness exponents \cite{FisherHuse1988,McMillan1984,BrayMoore1987,HartmannYoung2001,Amoruso2006}.

For the standard low-temperature Gaussian EA model on ordinary squares, the remaining problem can now be stated sharply:
\begin{quote}
Prove the exponential boundary-mixing condition, or a weaker sufficient condition for boundary localization, for the relevant half-space Gibbs states, and determine whether the resulting half-space pressure is independent of the exhaustion.
\end{quote}
The finite-volume interpolation identity, Gaussian concentration, and compactness yield tightness and subsequential information, but they do not supply the missing localization estimate. Theorem \ref{thm:exp-localization} proves that exponential boundary mixing would give the explicit rate \cref{eq:rho-explicit}; it does not establish that hypothesis for the ordinary low-temperature Gaussian EA model. This is why the present result should be read as a reduction and a rigorous strip theorem, not as a solution of the full cubic problem.

\section{Discussion}

The main theorem is intentionally formulated for a single face. Deriving a full fixed-boundary correction from face contributions requires a separate telescoping argument and quantitative control of interactions between faces. The strip construction avoids assuming such a decomposition and makes the tangential ergodic mechanism explicit.

There are several natural next steps. One is to verify \cref{eq:EBM}, or to replace it by a block or annealed mixing condition weaker than Dobrushin uniqueness, perhaps by a metastate construction or a quantitative estimate on the covariance between a boundary spin and degrees of freedom at depth $L$. Another is to establish a central limit theorem for the centered strip correction at fixed thickness. A third is to determine whether the half-space response is selection-independent in the two-dimensional Gaussian model. Each of these questions is genuinely about boundary states; none follows from finite-volume differentiation alone.

The assumptions on the enforced boundary layer are compatible with the standard fixed-exterior formulation. Fixing the exterior spins turns every crossing bond into a random one-site boundary field. The strip theorem applies to one face of that layer; a full cube requires additional face-additivity estimates that are deliberately not asserted here. The variables in the present paper are coupled across all sizes on one probability space, so the almost-sure claim in \cref{eq:stripas} has a precise meaning.

\Addresses

\begin{thebibliography}{99}
\footnotesize
\setlength{\itemsep}{0.08em}

\bibitem{EdwardsAnderson1975}
S.~F. Edwards and P.~W. Anderson, \emph{Theory of spin glasses},
J. Phys. F: Met. Phys. \textbf{5} (1975), 965--974.
\url{https://doi.org/10.1088/0305-4608/5/5/017}.

\bibitem{SherringtonKirkpatrick1975}
D.~Sherrington and S.~Kirkpatrick, \emph{Solvable model of a spin-glass},
Phys. Rev. Lett. \textbf{35} (1975), 1792--1796.
\url{https://doi.org/10.1103/PhysRevLett.35.1792}.

\bibitem{Parisi1979}
G.~Parisi, \emph{Infinite number of order parameters for spin-glasses},
Phys. Rev. Lett. \textbf{43} (1979), 1754--1756.
\url{https://doi.org/10.1103/PhysRevLett.43.1754}.

\bibitem{McMillan1984}
W.~L. McMillan, \emph{Domain-wall renormalization-group study of the three-dimensional random Ising model},
Phys. Rev. B \textbf{30} (1984), 476--479.
\url{https://doi.org/10.1103/PhysRevB.30.476}.

\bibitem{Imbrie1984}
J.~Z. Imbrie, \emph{Lower critical dimension of the random-field Ising model},
Phys. Rev. Lett. \textbf{53} (1984), 1747--1750.
\url{https://doi.org/10.1103/PhysRevLett.53.1747}.

\bibitem{Chayes1986}
J.~T. Chayes, L.~Chayes, D.~S. Fisher, and T.~Spencer,
\emph{Correlation length bounds for disordered Ising ferromagnets},
Phys. Rev. Lett. \textbf{57} (1986), 2999--3002.
\url{https://doi.org/10.1103/PhysRevLett.57.2999}.

\bibitem{BrayMoore1987}
A.~J. Bray and M.~A. Moore, \emph{Scaling theory of the spin-glass ground state},
J. Phys. A: Math. Gen. \textbf{20} (1987), L927--L933.
\url{https://doi.org/10.1088/0305-4470/20/15/026}.

\bibitem{BricmontKupiainen1987}
J.~Bricmont and A.~Kupiainen, \emph{Lower critical dimension for the random-field Ising model},
Phys. Rev. Lett. \textbf{59} (1987), 1829--1832.
\url{https://doi.org/10.1103/PhysRevLett.59.1829}.

\bibitem{FisherHuse1988}
D.~S. Fisher and D.~A. Huse, \emph{Equilibrium behavior of the spin-glass ordered phase},
Phys. Rev. B \textbf{38} (1988), 386--411.
\url{https://doi.org/10.1103/PhysRevB.38.386}.

\bibitem{AizenmanWehr1990}
M.~Aizenman and J.~Wehr, \emph{Rounding effects of quenched randomness on first-order phase transitions},
Commun. Math. Phys. \textbf{130} (1990), 489--528.
\url{https://doi.org/10.1007/BF02096933}.

\bibitem{HartmannYoung2001}
A.~K. Hartmann and A.~P. Young, \emph{Lower critical dimension of Ising spin glasses},
Phys. Rev. B \textbf{64} (2001), 180404.
\url{https://doi.org/10.1103/PhysRevB.64.180404}.

\bibitem{GuerraToninelli2002}
F.~Guerra and F.~L. Toninelli, \emph{The thermodynamic limit in mean field spin glass models},
Commun. Math. Phys. \textbf{230} (2002), 71--79.
\url{https://doi.org/10.1007/s00220-002-0699-y}.

\bibitem{NewmanStein2003}
C.~M. Newman and D.~L. Stein, \emph{Ordering and broken symmetry in short-ranged spin glasses},
J. Phys.: Condens. Matter \textbf{15} (2003), R1319--R1364.
\url{https://doi.org/10.1088/0953-8984/15/32/202}.

\bibitem{Boettcher2005}
S.~Boettcher, \emph{Stiffness of the Edwards-Anderson model in all dimensions},
Phys. Rev. B \textbf{72} (2005), 180405.
\url{https://doi.org/10.1103/PhysRevB.72.180405}.

\bibitem{Amoruso2006}
C.~Amoruso, A.~K. Hartmann, and A.~P. Young,
\emph{Stiffness exponent of two-dimensional Ising spin glasses},
Phys. Rev. Lett. \textbf{97} (2006), 087202.
\url{https://doi.org/10.1103/PhysRevLett.97.087202}.

\bibitem{ArguinNewmanSteinWehr2014}
L.-P. Arguin, C.~M. Newman, D.~L. Stein, and J.~Wehr,
\emph{Fluctuation bounds for interface free energies in spin glasses},
J. Stat. Phys. \textbf{156} (2014), 221--238.
\url{https://doi.org/10.1007/s10955-014-1009-8}.

\bibitem{Talagrand1996}
M.~Talagrand, \emph{A new look at independence in product spaces},
Invent. Math. \textbf{126} (1996), 505--530.
\url{https://doi.org/10.1007/s002220050107}.

\bibitem{Ledoux2001}
M.~Ledoux, \emph{The Concentration of Measure Phenomenon},
Mathematical Surveys and Monographs, vol.~89, Amer. Math. Soc., 2001.
\url{https://doi.org/10.1090/surv/089}.

\bibitem{BoucheronLugosiMassart2013}
S.~Boucheron, G.~Lugosi, and P.~Massart,
\emph{Concentration Inequalities: A Nonasymptotic Theory of Independence},
Oxford University Press, 2013.
\url{https://doi.org/10.1093/acprof:oso/9780199535255.001.0001}.

\bibitem{Bolthausen1982}
E.~Bolthausen, \emph{On the central limit theorem for stationary mixing random fields},
Ann. Probab. \textbf{10} (1982), 1047--1050.
\url{https://doi.org/10.1214/aop/1176993766}.

\bibitem{Georgii2011}
H.-O. Georgii, \emph{Gibbs Measures and Phase Transitions}, 2nd ed.,
de Gruyter, Berlin, 2011.
\url{https://doi.org/10.1515/9783110875550}.

\bibitem{Ruelle1969}
D.~Ruelle, \emph{Statistical Mechanics: Rigorous Results},
W.~A. Benjamin, 1969.
\url{https://doi.org/10.1007/978-3-642-08729-7}.

\bibitem{NewmanStein1997}
C.~M. Newman and D.~L. Stein, \emph{Metastate approach to thermodynamic chaos},
Phys. Rev. E \textbf{55} (1997), 5194--5211.
\url{https://doi.org/10.1103/PhysRevE.55.5194}.

\bibitem{AizenmanLebowitzRuelle1987}
M.~Aizenman, J.~L. Lebowitz, and D.~Ruelle,
\emph{Some rigorous results on the Sherrington-Kirkpatrick spin glass model},
Commun. Math. Phys. \textbf{112} (1987), 3--20.
\url{https://doi.org/10.1007/BF01218484}.

\bibitem{WangZhuChueng2026}
H.~Wang, K.~Zhu, and M.~Chueng,
\emph{Boundary Free Energies in Disordered Ising Models},
arXiv:2607.18326 (2026).
\url{https://arxiv.org/abs/2607.18326}.

\end{thebibliography}
\end{document}